\documentclass[11pt]{article}

\usepackage{fullpage}

\usepackage{times}
\usepackage{helvet}
\usepackage{courier}
\usepackage{amssymb,amsmath}
\usepackage{comment}
\usepackage{graphicx}

\usepackage{stmaryrd}

\usepackage{amsmath}
\usepackage{comment}
\usepackage{amssymb}

\newcommand{\ignore}[1]{}

\begin{document}

\title{An Algebraic Hardness Criterion for Surjective Constraint Satisfaction}

\author{Hubie Chen\\
Departamento LSI\\
Facultad de Inform\'{a}tica\\
Universidad del Pa\'{i}s Vasco\\
San Sebasti\'{a}n, Spain\\
\emph{and}\\
 IKERBASQUE, Basque Foundation for Science\\
Bilbao, Spain\\}

\date{ } 

\maketitle

\newcommand{\confversion}[1]{}
\newcommand{\longversion}[1]{#1}


\newtheorem{theorem}{Theorem}[section]
\newtheorem{conjecture}[theorem]{Conjecture}
\newtheorem{corollary}[theorem]{Corollary}
\newtheorem{proposition}[theorem]{Proposition}
\newtheorem{prop}[theorem]{Proposition}
\newtheorem{lemma}[theorem]{Lemma}
\newtheorem{remarkcore}[theorem]{Remark}
\newtheorem{exercisecore}[theorem]{Exercise}
\newtheorem{examplecore}[theorem]{Example}
\newtheorem{assumptioncore}[theorem]{Assumption}

\newenvironment{assumption}
  {\begin{assumptioncore}\rm}
  {\hfill $\Box$\end{assumptioncore}}

\newenvironment{example}
  {\begin{examplecore}\rm}
  {\hfill $\Box$\end{examplecore}}

\newenvironment{exercise}
 {\begin{exercisecore}\rm}
 {\hfill $\Box$\end{exercisecore}}

\newenvironment{remark}
 {\begin{remarkcore}\rm}
 {\hfill $\Box$\end{remarkcore}}

\newenvironment{proof}{\noindent\textbf{Proof\/}.}{$\Box$ \vspace{1mm}}

\newtheorem{researchq}{Research Question}

\newtheorem{newremarkcore}[theorem]{Remark}

\newenvironment{newremark}
  {\begin{newremarkcore}\rm}
  {\end{newremarkcore}}

\newtheorem{definitioncore}[theorem]{Definition}

\newenvironment{definition}
  {\begin{definitioncore}\rm}
  {\end{definitioncore}}

\newcommand{\ppequiv}{\mathsf{PPEQ}}
\newcommand{\eq}{\mathsf{EQ}}
\newcommand{\iso}{\mathsf{ISO}}
\newcommand{\ppeq}{\ppequiv}
\newcommand{\ppiso}{\mathsf{PPISO}}
\newcommand{\boolppiso}{\mathsf{BOOL}\mbox{-}\mathsf{PPISO}}
\newcommand{\csp}{\mathsf{CSP}}
\newcommand{\scsp}{\mathsf{SCSP}}
\newcommand{\gi}{\mathsf{GI}}
\newcommand{\ci}{\mathsf{CI}}
\newcommand{\rela}{\mathbf{A}}
\newcommand{\relb}{\mathbf{B}}
\newcommand{\relc}{\mathbf{C}}
\newcommand{\alga}{\mathbb{A}}
\newcommand{\algb}{\mathbb{B}}
\newcommand{\algab}{\mathbb{A}_{\relb}}

\newcommand{\idemp}{I}

\newcommand{\varv}{\mathcal{V}}
\newcommand{\variety}{\mathcal{V}}
\newcommand{\false}{\mathsf{false}}
\newcommand{\true}{\mathsf{true}}
\newcommand{\pol}{\mathsf{Pol}}
\newcommand{\inv}{\mathsf{Inv}}
\newcommand{\cc}{\mathcal{C}}
\newcommand{\alg}{\mathsf{Alg}}
\newcommand{\pitwo}{\Pi_2^p}
\newcommand{\sigmatwo}{\Sigma_2^p}
\newcommand{\pithree}{\Pi_3^p}
\newcommand{\sigmathree}{\Sigma_3^p}

\newcommand{\fancyg}{\mathcal{G}}
\newcommand{\tw}{\mathsf{tw}}

\newcommand{\mc}{\mathsf{MC}}
\newcommand{\mcs}{\mathsf{MC}_s}

\newcommand{\mcb}{\mathsf{MC_b}}

\newcommand{\qc}{\mathrm{QC}}
\newcommand{\normqc}{\mathrm{norm\mbox{-}QC}}
\newcommand{\rqc}{\mathsf{RQC\mbox{-}MC}}

\newcommand{\qcfo}{\mathrm{QCFO}}
\newcommand{\qcfofk}{\qcfo_{\forall}^k}
\newcommand{\qcfoek}{\qcfo_{\exists}^k}

\newcommand{\fo}{\mathrm{FO}}
\newcommand{\fok}{\mathrm{FO}^k}

\newcommand{\tup}[1]{\overline{#1}}

\newcommand{\nn}{\mathsf{nn}}
\newcommand{\bush}{\mathsf{bush}}
\newcommand{\width}{\mathsf{width}}

\newcommand{\un}{N^{\forall}}
\newcommand{\en}{N^{\exists}}

\newcommand{\ord}{\tup{u}}
\newcommand{\ordp}[1]{\tup{#1}}

\newcommand{\gc}{G^{-C}}

\newcommand{\ar}{\mathrm{ar}}
\newcommand{\free}{\mathsf{free}}
\newcommand{\vars}{\mathsf{vars}}

\newcommand{\qed}{}

\newcommand{\f}{\mathcal{F}}

\newcommand{\pow}{\wp}

\newcommand{\N}{\mathbb{N}}

\newcommand{\param}[1]{\mathsf{param}\textup{-}#1}

\newcommand{\dom}{\mathsf{dom}}

\newcommand{\org}{\mathrm{org}^+}
\newcommand{\lay}{\mathrm{lay}^+}

\newcommand{\und}[1]{\underline{#1}}

\newcommand{\clo}{\mathsf{closure}}

\newcommand{\thick}{\mathsf{thick}}
\newcommand{\thickl}{\thick_l}
\newcommand{\localthickl}{\mathsf{local}\textup{-}\thickl}
\newcommand{\quantthickl}{\mathsf{quant}\textup{-}\thickl}

\newcommand{\lowerdeg}{\mathsf{lower}\textup{-}\mathsf{deg}}

\newcommand{\restrict}{\upharpoonright}

\renewcommand{\nu}[1]{\textup{{\small $\mathsf{nu}$}-{$ #1 $}}}

\newcommand{\case}[1]{\textup{{\small $\mathsf{case}$}-{$ #1 $}}}

\newcommand{\coclique}{\textup{{\small  $\mathsf{co}$}-{\small $\mathsf{CLIQUE}$}}}
\newcommand{\clique}{\textup{{\small $\mathsf{CLIQUE}$}}}

\newcommand{\caseclique}{\case{\clique}}
\newcommand{\casecoclique}{\case{\coclique}}

\newcommand{\fpt}{\textup{\small $\mathsf{FPT}$}}
\newcommand{\wone}{\textup{\small $\mathsf{W[1]}$}}
\newcommand{\cowone}{\textup{\small $\mathsf{co}$-$\mathsf{W[1]}$}}

\renewcommand{\S}{\mathcal{S}}
\newcommand{\G}{\mathcal{G}}

\newcommand{\image}{\mathsf{image}}

\begin{abstract}
\begin{quote}
The constraint satisfaction problem (CSP)
on a
relational structure $\relb$ is to decide,
given a set of constraints on variables where the relations 
come from $\relb$, whether or not there is a assignment
to the variables satisfying all of the constraints;
the surjective CSP is the variant where one decides
the existence of a surjective satisfying assignment
onto the universe of $\relb$.
We present  an algebraic condition on the polymorphism clone
of $\relb$
and prove that it is sufficient for
the hardness of the surjective CSP on a finite structure $\relb$,
in the sense that
this problem admits a reduction from a certain
fixed-structure CSP.
To our knowledge, this is the first result
that allows one to use algebraic information 
from a relational structure $\relb$
to infer information on the complexity hardness of
surjective constraint satisfaction on $\relb$.
A corollary of our result is that, on any finite non-trivial structure 
having only essentially unary polymorphisms, 
surjective constraint satisfaction is NP-complete.
\end{quote}
\end{abstract}

\section{Introduction} \label{section:introduction}

The \emph{constraint satisfaction problem (CSP)} is 
a computational problem in which one is to decide,
given a set of constraints on variables, whether or not
there is an assignment to the variables satisfying all of the
constraints.
This problem appears in many guises throughout computer science,
for instance, in database theory, artificial intelligence, and
the study of graph homomorphisms.
One obtains a rich and natural family of problems by
defining, for each relational structure $\relb$,
the problem $\csp(\relb)$ to be the case of the CSP where
the relations used to specify constraints must come from $\relb$.
An increasing literature studies the algorithmic and complexity
behavior of this problem family, focusing on finite and finite-like
structures~\cite{BartoKozik09-boundedwidth,IMMVW10-tractabilityfewsubpowers,Bodirsky12-hab}; 
a primary research issue is to
determine which such problems are polynomial-time tractable,
and which are not.
To this end of classifying problems, a so-called \emph{algebraic approach}
has been quite fruitful~\cite{BulatovJeavonsKrokhin05-finitealgebras}.  In short, this approach is founded
on the facts that the complexity of a problem $\csp(\relb)$
depends (up to polynomial-time reducibility)
only on the set of relations that are 
\emph{primitive positive definable} from $\relb$,
and that this set of relations can be derived from the 
\emph{clone of polymorphisms} of $\relb$.  
Hence, the project of classifying all relational structures
according to the complexity of $\csp(\relb)$ can be formulated
as a classification question on clones;
this permits the employment of algebraic notions and techniques
in this project.
(See the next section for formal definitions of the notions discussed
in this introduction.)

A natural variant of the CSP is the \emph{surjective CSP},
where an instance is again a set of constraints,
but one is to decide whether or not there is a \emph{surjective}
satisfying assignment to the variables.
For each relational structure $\relb$, 
one may define $\scsp(\relb)$
to be the surjective CSP on $\relb$, in analogy to the definition
of $\csp(\relb)$.
Note that one can equivalently define $\scsp(\relb)$
to be the problem of deciding, given as input a relational structure
$\rela$,
whether or not there is a surjective homomorphism from $\rela$ to $\relb$.
An early result on this problem family was the complexity classification
of all two-element structures~\cite[Proposition
6.11]{CreignouKhannaSudan01-boolean},
~\cite[Proposition 4.7]{CreignouHebrard97-generatingall}.
There is recent interest in understanding the complexity of
these problems, which perhaps focuses on the cases where 
the structure $\relb$ is a graph; we refer the reader to the 
survey~\cite{BodirskyKaraMartin12-surjectivesurvey}
for further information and pointers, and also 
can reference the related articles~\cite{Uppman12-surcsptwoelts,GolovachPaulusmaSong12-partreflexivetrees,Hell14-prescribedpatterns}.
The introduction in the survey~\cite{BodirskyKaraMartin12-surjectivesurvey} suggests that
the problems $\scsp(\relb)$ ``seem to be very difficult to classify in
terms of complexity'', and that ``standard methods to prove easiness
or hardness fail.''
Indeed, in contrast to the vanilla CSP,
there is no known way to reduce the complexity classification of the problems
$\scsp(\relb)$ to a classification of clones.
In particular, there is no known result showing that the complexity of 
a problem $\scsp(\relb)$ depends only on the relations that are 
primitive positive definable from $\relb$.
Thus far, there has been no success in using algebraic information
based on the polymorphisms of $\relb$ 
to deduce complexity hardness consequences
for the problem $\scsp(\relb)$.
 (The claims given here are relative to the best of our knowledge).

In this article, we give (to our knowledge) the first result
which allows one to use algebraic information from
the polymorphisms of a structure $\relb$ to infer information about the
complexity hardness of $\scsp(\relb)$.  
Let us assume that the structures under discussion are finite
relational structures.
It is known and straightforward to verify that
the problem $\scsp(\relb)$ polynomial-time reduces to
the problem $\csp(\relb^+)$, where $\relb^+$ denotes the expansion
of $\relb$ by constants~\cite[Section 2]{BodirskyKaraMartin12-surjectivesurvey}.  
We give a sufficient condition for the problem $\csp(\relb^+)$
to polynomial-time reduce to the problem $\scsp(\relb)$, and hence
for the equivalence of these two problems
(up to polynomial-time reducibility).
From a high level, our sufficient condition 
requires a certain relationship 
between the diagonal and the image of an operation, 
for each operation in the polymorphism clone of $\relb$.
Any structure $\relb$ whose polymorphisms are all essentially unary
satisfies our sufficient condition, and a corollary of our main
theorem
is that, for any such structure $\relb$
(having a non-trivial universe), the problem $\scsp(\relb)$
is NP-complete.
In the classification of two-element structures~\cite[Proposition
6.11]{CreignouKhannaSudan01-boolean}, 
each structure
on which $\scsp(\relb)$ is proved NP-complete has only 
essentially unary polymorphisms (this can be inferred 
from existing results~\cite[Theorem 5.1]{Chen09-Rendezvous}).
Hence, the just-named corollary
yields a new algebraic proof of the hardness results needed for this
classification;
we find this proof to be a desirable, concise alternative to the 
relational argumentation carried out in previously known proofs of
this classification~\cite[Proposition
6.11]{CreignouKhannaSudan01-boolean},
~\cite[Proposition 4.7]{CreignouHebrard97-generatingall}.

We hope that our result might lead to further interaction between 
the study of surjective constraint satisfaction and universal algebra,
and in particular that the techniques that we present might be used
to prove new hardness results or to simplify known hardness proofs.

\section{Preliminaries}

For a natural number $n$, we use $\und{n}$ to denote the set
$\{ 1, \ldots, n \}$.
We use $\pow(B)$ to denote the power set of a set $B$.

\subsection{Logic and computational problems}

We make basic use of the syntax and semantics of 
relational first-order logic.
A \emph{signature} is a set of \emph{relation symbols};
each relation symbol $R$ has an associated arity (a natural number),
denoted by $\ar(R)$.
A \emph{structure} $\relb$ over signature $\sigma$
consists of a \emph{universe} $B$ which is a set,
and an interpretation $R^{\relb} \subseteq B^{\ar(R)}$
for each relation symbol $R \in \sigma$.
In this article, we 
assume that signatures under discussion are finite, and 
focus on finite structures; a structure is finite
if its universe is finite.
When $\relb$ is a structure over signature $\sigma$, we define $\relb^+$ to
be the expansion of $\relb$ ``by constants'',
that is, the expansion 
which is defined on signature 
$\sigma \cup \{ C_b ~|~ b \in B \}$,
where each $C_b$ has unary arity and is assumed not to be in $\sigma$,
and where $C_b^{\relb^+} = \{ b \}$.

By an \emph{atom}, we refer to a formula of the form
$R(v_1, \ldots, v_k)$ where $R$ is a relation symbol, $k = \ar(R)$,
and the $v_i$ are variables;
by a \emph{variable equality}, we refer to a formula of the form
$u = v$ where $u$ and $v$ are variables.
A \emph{pp-formula} (short for \emph{primitive positive formula})
is a formula built using atoms, variable equalities, 
conjunction $(\wedge)$, and existential quantification $(\exists)$.
A \emph{quantifier-free pp-formula} is a pp-formula that does not
contain existential quantification, that is, a pp-formula
that is a conjunction of atoms and variable equalities.
A relation $P \subseteq B^m$ 
is \emph{pp-definable} over a structure $\relb$
if there exists a pp-formula $\psi(x_1, \ldots, x_m)$
such that a tuple $(b_1, \ldots, b_m)$ is in $P$
if and only if $\relb, b_1, \ldots, b_m \models \psi$;
when such a pp-formula exists, it is called 
a \emph{pp-definition} of $P$ over $\relb$.

We now define the computational problems to be studied.
For each structure $\relb$, define
$\csp(\relb)$ to be the problem of deciding, given 
a conjunction $\phi$ of atoms (over the signature of $\relb$),
whether or not there is a map $f$ to $B$
defined on the variables of $\phi$
such that $\relb, f \models \phi$.
For each structure $\relb$, define
$\scsp(\relb)$ to be the problem of deciding, given 
a pair $(U, \phi)$ where $U$ is a set of variables and
$\phi$ is a conjunction of atoms (over the signature of $\relb$)
with variables from $U$, whether or not there is a surjective map
$f: U \to B$ 
such that $\relb, f \models \phi$.

Note that these two problems are sometimes formulated
as relational homomorphism problems;
for example,
one can define  $\scsp(\relb)$ as the problem of deciding,
given a structure $\rela$ over the signature of $\relb$,
whether or not there is a surjective homomorphism from $\rela$ to
$\relb$.
This is an equivalent formulation:
an instance $(U, \phi)$
of $\scsp(\relb)$ can be translated naturally to the structure $\rela$
with universe $U$ and where $(u_1, \ldots, u_k) \in R^{\rela}$
if and only if $R(u_1, \ldots, u_k)$ is present in $\phi$;
this structure $\rela$ admits a surjective homomorphism to $\relb$
if and only if $(U, \phi)$ is a \emph{yes} instance of $\scsp(\relb)$
as we have defined it.
One can also naturally invert this passage, to translate
from the homomorphism formulation to ours.
Let us remark that
in our formulation of $\scsp(\relb)$,
when $(U, \phi)$ is an instance, it is permitted that $U$ contain
variables that are not present in $\phi$; indeed, whether or not
the instance is a \emph{yes} instance may be sensitive to the
exact number of such variables, and this is why 
this variable set is given explicitly.

We now make a simple observation which essentially says that
one could alternatively define $\scsp(\relb)$
by allowing the formula $\phi$ to be a quantifier-free pp-formula,
as variable equalities may be efficiently eliminated
in a way that preserves the existence of a surjective satisfying assignment.

\begin{prop}
\label{prop:remove-equalities}
There exists a polynomial-time algorithm that,
given a pair $(W, \phi)$ where 
$\phi$ is a quantifier-free pp-formula with variables from $W$,
outputs a pair $(W', \phi')$ where 
$\phi'$ is a conjunction of atoms with variables from $W'$
and having the following property:
 for any structure $\relb$ (whose signature contains
the relation symbols present in $\phi$),
there exists a surjective map $f: W \to B$ such that
$\relb, f \models \phi$
if and only if
there exists a surjective map $f': W' \to B$ such that
$\relb, f' \models \phi'$.
\end{prop}

\begin{proof}
The algorithm repeatedly eliminates variable equalities one at a time, until
no more exist.
Precisely, given a pair $(W, \phi)$,
it iterates the following two steps as long as $\phi$ contains a
variable equality.
The first step is to simply obtain $\phi'$ 
by removing from $\phi$ all variable equalities $u = u$
that equate the same variable, and then replace $(W, \phi)$
by $(W, \phi')$.
The second step is to check if $\phi$ contains a variable equality
$u = v$ between two different variables; if so, 
the algorithm 
picks such an equality $u = v$, 
obtains $\phi'$ by replacing all instances of $v$ with $u$,
and then replaces $(W, \phi)$ by $(W \setminus \{ v \}, \phi')$.
The output of the algorithm is the final value of $(W, \phi)$.
It is straightforwardly verified that this final value has the desired
property (by checking that each of the two steps preserve the property).
\end{proof}

\subsection{Algebra}

All operations under consideration are assumed to be
of finite arity greater than or equal to $1$.
We use $\image(f)$ to denote the image of an operation $f$.
The \emph{diagonal} of an operation $f: B^k \to B$, 
denoted by $\hat{f}$,
is the unary operation defined by $\hat{f}(b) = f(b, \ldots, b)$.
Although not the usual definition, it is correct to say that an
operation $f: B^k \to B$
is \emph{essentially unary} if and only if
there exists $i \in \und{k}$ such that
$f(b_1, \ldots, b_k) = \hat{f}(b_i)$.

When $t^1, \ldots, t^k$ are tuples on $B$ having the same arity $m$
and $f: B^k \to B$ is an operation,
the tuple $f(t^1, \ldots, t^k)$ is the arity $m$ tuple obtained by applying
$f$ coordinatewise.
The entries of a tuple $t$ of arity $m$ are denoted by $t = (t_1, \ldots, t_m)$.
Let $P \subseteq B^m$ be a relation, and let $f: B^k \to B$ be an
operation;
we say that \emph{$f$ is a polymorphism of $P$}
or that \emph{$P$ is preserved by $f$}
if for any choice of $k$ tuples $t^1, \ldots, t^k \in P$,
it holds that $f(t^1, \ldots, t^k) \in P$.
An operation $f: B^k \to B$ is a \emph{polymorphism} of
a structure $\relb$ if $f$ is a polymorphism of 
each relation of $\relb$; 
we use $\pol(\relb)$ to denote the set of all polymorphisms of
$\relb$.
It is known that, for any structure $\relb$,
the set $\pol(\relb)$ is a \emph{clone}, which is a set of operations
that contains all projections and is closed under composition.

We will make use of the following characterization of pp-definability
relative to a structure $\relb$.

\begin{theorem} \cite{Geiger68-closed,BKKR69-galois}
\label{thm:galois}
A non-empty relation $P \subseteq B^m$ is pp-definable
over a finite structure $\relb$ if and only if
each operation $f \in \pol(\relb)$ is a polymorphism of $P$.
\end{theorem}

\section{Hardness result}

Throughout this section, $B$ will be a finite set;
we set $n = |B|$ and use $b_1^*, \ldots, b_n^*$ to denote
a fixed enumeration of the elements of $B$.

We give a complexity hardness result on $\scsp(\relb)$ 
under the assumption that the polymorphism clone of $\relb$ 
satisfies a particular property, which we now define.
We say that a clone $C$ on a set $B$ is \emph{diagonal-cautious} 
if there exists a map $G: B^n \to \pow(B)$
such that:
\begin{itemize}
\item for each operation $f \in C$,
it holds that $\image(f) \subseteq G(\hat{f}(b_1^*), \ldots, \hat{f}(b_n^*))$, and
\item for each tuple $(b_1, \ldots, b_n) \in B^n$,
if $\{ b_1, \ldots, b_n \} \neq B$, then
$G(b_1, \ldots, b_n) \neq B$.
\end{itemize}
Roughly speaking, this condition yields that, when the diagonal of an operation
$f \in C$ is not surjective, then the image of $f$ is contained
in a proper subset of $B$ that is given by $G$ as a function of
$\hat{f}$.

\begin{example}
\label{ex:essentially-unary}
When a clone consists only of essentially unary operations,
it is diagonal-cautious via the map
$G(b_1, \ldots, b_n) = \{ b_1, \ldots, b_n \}$,
as for an essentially unary operation $f$, it holds that
$\image(f) \subseteq \{ \hat{f}(b_1^*), \ldots, \hat{f}(b_n^*) \} = \image(\hat{f})$.
\end{example}

\begin{example}
When each operation in a clone has a surjective diagonal,
the clone is diagonal-cautious via the map $G$ given in the
previous example.
\end{example}

The following lemma is the key to our hardness result;
it provides a \emph{quantifier-free} pp-formula which will be used
as a gadget in the hardness proof.

\begin{lemma}
\label{lemma:qf}
Suppose that $\relb$ is a finite structure whose universe $B$
has size strictly greater than $1$,
and suppose that $\pol(\relb)$ is diagonal-cautious via $G$.
There exists a quantifier-free pp-formula 
$\psi(v_1, \ldots, v_n, x, y_1, \ldots, y_m)$
such that:
\begin{itemize}

\item[(1)] If it holds that
$\relb, b_1, \ldots, b_n, c, d_1, \ldots, d_m \models \psi$,
then  $b_1, \ldots, b_n, c, d_1, \ldots, d_m \in G(b_1, \ldots, b_n)$.

\item[(2)] For each $c \in B$, it holds that
$\relb, b_1^*, \ldots, b_n^*, c \models \exists y_1 \ldots \exists y_m
\psi$.

\item[(3)] If it holds that $\relb, b_1, \ldots, b_n \models
\exists x \exists y_1 \ldots \exists y_m \psi$,
then there exists a unary polymorphism $u$ of $\relb$
such that $(u(b_1^*), \ldots, u(b_n^*)) = (b_1, \ldots, b_n)$.

\end{itemize}
\end{lemma}

\begin{proof}
Let

\begin{center}
\( \begin{array}{ccc}
t^1 & = & (t^1_1, \ldots, t^1_{n^n}) \\
\vdots & & \vdots \\
t^n & = & (t^n_1, \ldots, t^n_{n^n}) \\
\end{array} \)
\end{center}
be tuples from $B^{(n^n)}$ such that the following three conditions
hold:
\begin{itemize}

\item[($\alpha$)]  
It holds that
$\{ (t^1_i, \ldots, t^n_i) ~|~ i \in \und{n^n} \} = B^n$.

\item[($\beta$)] For each $i \in \und{n}$, it holds that
$\{ t^1_i, \ldots, t^n_i \} = \{ b_i^* \}$.

\item[($\gamma$)] It holds that $\{ t^1_{n+1}, \ldots, t^n_{n+1} \} = B$.

\end{itemize}
Visualizing the tuples as rows (as above), condition $(\alpha)$ is
equivalent to the assertion that
each tuple from $B^n$ occurs exactly once as a column;
condition $(\beta)$ enforces that the first $n$ columns are
the tuples with constant values $b_1^*, \ldots, b_n^*$ (respectively);
and, condition $(\gamma)$ enforces that the $(n+1)$th column
is a \emph{rainbow} column in that each element of $B$ occurs
exactly once in that column.

Let $P$ be the $(n^n)$-ary relation
$\{ f(t^1, \ldots, t^n) ~|~ \textup{ $f$ is an $n$-ary polymorphism of
  $\relb$ } \}$.
It is well-known and straightforward to verify that 
the relation $P$ is preserved by all polymorphisms of $\relb$.
By Theorem~\ref{thm:galois},
we have that $P$ has a pp-definition
$\phi(w_1, \ldots, w_{n^n})$ over $\relb$.
We may and do assume that
$\phi$ is in prenex normal form, in particular, we assume
$\phi = \exists z_1 \ldots \exists z_q \theta(w_1, \ldots, w_{n^n},
z_1, \ldots, z_q)$ where $\theta$ is a conjunction of atoms and
equalities.

Since $t^1, \ldots, t^n \in P$, there exist tuples 
$u^1, \ldots, u^n \in B^q$
such that, for each $k \in \und{n}$, 
it holds that $\relb, (t^k, u^k) \models \theta$.
By condition $(\alpha)$, there exist values
$a_1, \ldots, a_q \in \und{n^n}$
such that, for each $i \in \und{q}$,
it holds that
$(u^1_i, \ldots, u_i^n) = (t^1_{a_i}, \ldots, t^n_{a_i})$.
Define $\psi(w_1, \ldots, w_{n^n})$ as
$\theta(w_1, \ldots, w_{n^n}, w_{a_1}, \ldots, w_{a_q})$.
We associate the variable tuples $(w_1, \ldots, w_{n^n})$
and $(v_1, \ldots, v_n, x, y_1, \ldots, y_m)$, 
so that
$\psi$ may be viewed as a formula with variables from 
$\{ v_1, \ldots, v_n, x, y_1, \ldots, y_m \}$.
We verify that $\psi$ has the three conditions given in the lemma
statement,
as follows.

(1): Suppose that $\relb, b_1, \ldots, b_n, c, d_1, \ldots, d_m
\models \psi$.  Then
$(b_1, \ldots, b_n, c, d_1, \ldots, d_m)$
is of the form $f(t^1, \ldots, t^n)$ where $f$ is a polymorphism of
$\relb$.
We have
$$\{ b_1, \ldots, b_n, c, d_1, \ldots, d_m \} \subseteq \image(f)
\subseteq G(\hat{f}(b_1^*), \ldots, \hat{f}(b_n^*))
= G(b_1, \ldots, b_n).$$
The second containment follows from the definition of 
\emph{diagonal-cautious}, and the equality follows from $(\beta)$.

(2): We had that, for each $k \in \und{n}$, 
it holds that $\relb, (t^k, u^k) \models \theta$.
By the choice of the $a_i$ and the definition of $\psi$,
it holds (for each $k \in \und{n}$) that $\relb, t^k \models \psi$.
Condition (2) then follows immediately from conditions $(\alpha)$ and $(\beta)$.

(3): Suppose that $\relb, b_1, \ldots, b_n \models 
\exists x \exists y_1 \ldots \exists y_m \psi$.
By definition of $\psi$, we have that there exists a tuple
beginning with $(b_1, \ldots, b_n)$ that satisfies $\theta$ on
$\relb$.
By the definition of $\theta$, we have that there exists a tuple $t$
beginning with $(b_1, \ldots, b_n)$ such that $t \in P$.
There exists a polymorphism $f$ of $\relb$ such that
$t = f(t^1, \ldots, t^n)$.
By condition $(\beta)$, we have that
$(\hat{f}(b_1^*), \ldots, \hat{f}(b_n^*)) = (b_1, \ldots, b_n)$.
\end{proof}

Let us make some remarks.
The relation $P$ in the just-given proof
is straightforwardly verified
(via Theorem~\ref{thm:galois})
 to be the smallest pp-definable relation 
(over $\relb$)
that contains all of the tuples $t^1, \ldots, t^n$.   
The definition of $\psi$ yields that the relation defined by $\psi$
(over $\relb$)
is a subset of $P$; the verification of condition (2) yields
that each of the tuples $t^1, \ldots, t^n$ is contained in the
relation defined by $\psi$.
Therefore, the formula $\psi$ defines precisely the relation $P$.
A key feature of the lemma, which is critical for our application
to surjective constraint satisfaction, 
is that the formula $\psi$ is quantifier-free.
We believe that it may be of interest to search for 
further applications of this lemma.

The following is our main theorem.

\begin{theorem}
\label{thm:main}
Suppose that $\relb$ is a finite structure such that
$\pol(\relb)$ is diagonal-cautious.
Then the problem $\csp(\relb^+)$ many-one polynomial-time reduces to 
$\scsp(\relb)$.
\end{theorem}

\begin{proof}
The result is clear if the universe $B$ of $\relb$ has size $1$,
so assume that it has size strictly greater than $1$.
Let $\psi(v_1, \ldots, v_n, x, y_1, \ldots, y_m)$ 
be the quantifier-free pp-formula given by Lemma~\ref{lemma:qf}.
Let $\phi$ be an instance of $\csp(\relb^+)$ which uses variables $U$.
The reduction creates an instance 
of $\scsp(\relb)$ as follows.
It first creates a quantifier-free pp-formula $\phi'$ that uses variables
$$U' = U \cup \{ v_1, \ldots, v_n \} \cup \bigcup_{u \in U} \{ y^u_1,
\ldots, y^u_m \}.$$
Here, each of the variables given in the description of $U'$ is
assumed to be distinct from the others, so that
$|U'| = |U| + n + |U|m$.
Let $\phi^=$ be the formula obtained from $\phi$
by replacing each atom of the form $C_{b_j^*}(u)$ 
by the variable equality $u = v_j$.
The formula $\phi'$ is defined as
$\phi^= \wedge \bigwedge_{u \in U}
 \psi(v_1, \ldots, v_n, u, y_1^u, \ldots, y_m^u)$.
The output of the reduction is the algorithm of 
Proposition~\ref{prop:remove-equalities}
applied to $(U', \phi')$.

To prove the correctness of this reduction,
we need to show that 
there exists a map $f: U \to B$ such that
$\relb^+, f \models \phi$ if and only if
there exists a surjective map $f': U' \to B$
such that
$\relb, f' \models \phi'$.

For the forward direction, define 
$f^=: U \cup \{ v_1, \ldots, v_n \} \to B$
to be the extension of $f$ such that
$f^=(v_i) = b_i^*$ for each $i \in \und{n}$.
It holds that $f^=$ is surjective and 
that $\relb, f^= \models \phi^=$.
By property (2) in the statement of Lemma~\ref{lemma:qf},
there exists an extension $f': U' \to B$ of $f^=$ such that
$\relb, f' \models \phi'$.

For the backward direction, we argue as follows.
We claim that $\{ f'(v_1), \ldots, f'(v_n) \} = B$.
If not, then by 
the definition of diagonal-cautious, it holds that
$G(f'(v_1), \ldots, f'(v_n)) \neq B$; 
by property (1) in the statement of Lemma~\ref{lemma:qf}
and by the definition of $\phi'$,
it follows that $f'(u') \in G(f'(v_1), \ldots, f'(v_n))$
for each $u' \in U'$, contradicting that $f'$ is surjective.
By property (3) in the statement of Lemma~\ref{lemma:qf},
there exists a unary polymorphism $u$ of $\relb$
such that $(u(b_1^*), \ldots, u(b_n^*)) = (f'(v_1), \ldots, f'(v_n))$;
by the just-established claim, $u$ is a bijection.
Since the set of unary polymorphisms of a structure is closed under
composition 
and since
$B$ is by assumption finite, the inverse $u^{-1}$ of $u$ 
is also a polymorphism of $\relb$.
Hence it holds that $\relb, u^{-1}(f') \models \phi'$,
where $u^{-1}(f')$ denotes the composition of $f'$ with $u^{-1}$. 
Since $u^{-1}(f')$ maps each variable $v_j$ to $b_j^*$, we can infer that
$\relb^+, u^{-1}(f') \models \phi$.
\end{proof}

\begin{corollary}
Suppose that $\relb$ is a finite structure whose universe $B$
has size strictly greater than $1$.
If each polymorphism of $\relb$ is essentially unary,
then $\scsp(\relb)$ is NP-complete.
\end{corollary}

\begin{proof}
The problem $\scsp(\relb)$ is in NP whenever $\relb$ is a finite
structure,
so it suffices to prove NP-hardness.
By Example~\ref{ex:essentially-unary},
 we have that $\pol(\relb)$ is diagonal-cautious.
Hence, we can apply
Theorem~\ref{thm:main}, and it suffices to argue that
$\csp(\relb^+)$ is NP-hard.
Since $\relb^+$ is by definition the expansion of $\relb$ with constants,
the polymorphisms of $\relb^+$ are exactly the idempotent
polymorphisms
of $\relb$; here then, the polymorphisms of $\relb^+$ are
the projections.
It is well-known that a structure having only projections
as polymorphisms has a NP-hard CSP~\cite{BulatovJeavonsKrokhin05-finitealgebras}
(note that in this case, 
Theorem~\ref{thm:galois} yields that 
every relation over the structure's universe is pp-definable).
\end{proof}

\paragraph{Acknowledgements.} 
The author thanks Matt Valeriote, Barny Martin,
and Yuichi Yoshida for useful comments and feedback.
 The author was supported 
by the Spanish Project FORMALISM (TIN2007-66523), 
by the Basque Government Project S-PE12UN050(SAI12/219), and 
by the University of the Basque Country under grant UFI11/45.

\bibliographystyle{plain}

\bibliography{/Users/hubiec/Dropbox/active/writing/hubiebib.bib}

\end{document}